\documentclass[a4paper,unpublished,onecolumn,11pt]{quantumarticle}
\pdfoutput=1

\usepackage{geometry,amssymb,latexsym,amsmath,graphicx}

\usepackage{stmaryrd}
\usepackage[autostyle, english = american]{csquotes}
\usepackage{hyperref}
\usepackage[capitalize]{cleveref}
\usepackage{amsthm}
\usepackage{tikz}
\MakeOuterQuote{"}
\usepackage{parskip}
\usepackage{thm-restate}
\usepackage{enumitem}
\usepackage{float}

\definecolor{fillvertices}{RGB}{250,240,230}
\definecolor{cornflowerblue}{RGB}{100,149,237}

\usepackage{mathtools}

\newcommand{\commentaire}

\newcommand{\se}{\subseteq}
\newcommand{\ls}{\leqslant}
\newcommand{\gs}{\geqslant}
\newcommand{\sm}{\setminus}
\newcommand{\cutrk}{\textup{cutrk}}

\newcommand{\valuescale}{0.81}

\newtheoremstyle{mythmstyle} 
  {6pt}   
  {0pt}   
  {\itshape}  
  {}      
  {\bfseries} 
  {.}     
  {0.5em} 
  {}      

\theoremstyle{mythmstyle}

\newtheorem{theorem}{Theorem}
\newtheorem*{theorem*}{Theorem}
\newtheorem{definition}{Definition}
\newtheorem{lemma}{Lemma}
	
\newtheorem{proposition}{Proposition}
\newtheorem{corollary}{Corollary}

\usepackage{algorithm2e}
\RestyleAlgo{ruled}
\SetKwComment{Comment}{/* }{ */}
\usepackage{dsfont}

\newcommand{\NN}{\mathbb{N}}

\newcommand{\calM}{{\cal M }}

\begin{document}

\title{Covering a Graph with Minimal Local Sets	}
\author{Nathan Claudet}
\affiliation{Inria Mocqua, LORIA, CNRS, Universit\'e de Lorraine,  F-54000 Nancy, France}
\author{ Simon Perdrix}
\affiliation{Inria Mocqua, LORIA, CNRS, Universit\'e de Lorraine,  F-54000 Nancy, France}

\date{}
\maketitle

\begin{abstract}
    Local sets, a graph structure invariant  under local complementation,  have been originally introduced in the context of quantum computing for the study of quantum entanglement within the so-called graph state formalism. A local set in a graph is made of a non-empty set of vertices together with its odd neighborhood. We show that any graph can be covered by minimal local sets, i.e. that every vertex is contained in at least one local set that is minimal by inclusion. More precisely, we introduce an algorithm for finding a minimal local set cover in polynomial time. This result is proved by exploring the link between local sets and cut-rank. We prove some additional results on minimal local sets: we give tight bounds on their size, and we show that there can be exponentially many of them in a graph. Finally, we provide an extension of our definitions and our main result to $q$-multigraphs, the graphical counterpart of quantum qudit graph states.
\end{abstract}

\section{Introduction} 

\noindent{\bf Context and contributions.} Local complementation, introduced by  Kotzig \cite{Kotzig68}, is a  transformation that consists in complementing the neighborhood of a given vertex of a graph. Local complementation is strongly related to the notion of vertex minor \cite{OUM200579,kim2023vertex,DHW:howtotransform},  has various applications in the study particular families of graphs including circle graphs \cite{DEFRAYSSEIX198129,BOUCHET1994107}, and plays a central role in the study of graph parameters like rank-width \cite{OUM201715} and cut-rank \cite{Bouchet1989,OUM2006}. Local complementation finds significant applications in quantum computing, ranging from addressing fundamental questions on quantum entanglement \cite{VandenNest04,Hein06} to its utilization in various protocols \cite{Markham08,marin2013,claudet2023small,cautres2024vertexminor}, and computational models like ZX-calculus \cite{Coecke_2011,DP09,DP14,gong2017equivalence,DKPW} or measurement-based quantum computation \cite{mhalla2014graph,danos2009extended,mhalla:hal-00934104}.

We consider in this paper a structure invariant under local complementation called \emph{local set}  \cite{Perdrix06}.  Given a simple, undirected graph $G$, a non-empty set of vertices is said \emph{local} if it is of the form $D \cup Odd_G(D)$, where $Odd_G(D)$ is the set of vertices that share an odd number of edges with vertices of $D$. Minimal local sets are local sets that are minimal by inclusion.

Local sets are related to the local minimum degree $\delta_{loc}(G)$ of a graph $G$ which is  the minimum degree of its vertices, up to local complementation\footnote{i.e. $\delta_{loc}(G)=\min_{H \equiv G}\delta(H)$ where $H\equiv G$ if $H$ can be transformed in $G$ by a sequence of local complementations.}. It has been proved in \cite{Perdrix06} that the local minimum degree of a graph is the size of its smallest (minimal) local set minus one. The local minimum degree can also be characterized using the cut-rank function, which associates with any cut of a graph the rank of the matrix describing the edges of the cut.

Foliage and foliage partition \cite{Pappa2022,burchardt2023foliage,zhang2023bell} have been recently introduced: the foliage of a graph can be defined as the union of the size-2 minimal local sets;  and two vertices are in the same component of the foliage partition if they belong to a same minimal local set of size two. Notice that minimal local sets of size two are composed of two non-isolated vertices, which can be either a leaf and its neighbor or twins.  
In \cite{burchardt2023foliage}, the question of the extension of foliation is raised as  size-2 local sets do not cover all the vertices of a graph in general. One can naturally relax the size restriction, and hence consider minimal-local-set covers -- MLS covers for short. One of our main result is to show that any vertex can be covered by a minimal local set, and thus that any graph has an MLS cover. This result relies on our cut-rank-based characterisation of minimal local sets together with the peculiar properties of the cut-rank which, as a set function, is symmetric, linearly bounded and submodular.  We also introduce an efficient algorithm that, given a graph,  produces a family of  a minimal local sets that cover any vertex of the graph. Furthermore, we investigate properties related to the size and quantity of minimal local sets in a graph, showing in particular that some graphs have an exponential number of minimal local sets.

\smallskip

\noindent \textbf{Related work on  quantum graph states. } 
Graph states are resources for quantum computing, that are in one-to-one correspondence with mathematical graphs. The survey \cite{Hein06} by Hein et al. provides an excellent introduction to graph states.
Local sets were introduced in \cite{Perdrix06} as a combinatorial tool in the context of the preparation of graph states.

Within  the graph state formalism, two locally equivalent graphs\footnote{i.e. equal up to local complementation.} represent the same entanglement\footnote{More precisely the corresponding graph states are equal up to local unitaries.}, and two graphs representing the same entanglement have the same cut-rank functions. Notice that the converses of these two assertions are false\footnote{Two graphs representing the same entanglement does not generally imply that they are locally equivalent: a counterexample of order 27 has been discovered using computer assisted methods \cite{Ji07}. Moreover, two graphs having the same cut-rank function does not generally imply that they represent the same entanglement: one counter-example involves two isomorphic Petersen graphs \cite{Fon-Der-Flaass1996,Hein06}.}, and little is known on non locally-equivalent graphs that represent the same entanglement except that their  minimal local sets must satisfy some strong constraints \cite{VandenNest05, Zeng07}. 

The local minimum degree is strongly related to the notion of $k$-uniformity \cite{Raissi2022,Raissi2020,Klobus2019,Goyeneche2018,Goyeneche2014,Scott2004}, which is defined as follows: a quantum state is $k$-uniform if for any subset of $k$ qubits, the corresponding reduced state is maximally mixed, i.e. roughly speaking contains no information. Notice that a graph state is $k$-uniform if and only if the minimal local sets of the corresponding graph contain at least $k+1$ vertices\footnote{Local sets are nothing but the support of the stabilizers of the state, the reduced state on $k$ qubits \cite{Hein06} is a maximally mixed state if it does not contain any local set. }, i.e. its  local minimum degree is $k$ or more. The $\lfloor n/2 \rfloor$-uniform $n$-qubit states are called absolutely maximally entangled \cite{Helwig2013,Helwig2013existence,Helwig2012,Huber2018,Goyeneche2015,Facchi2008,Arnaud2013}. A graph state is absolutely maximally entangled if and only if the (minimal) local sets of the corresponding graph are of size at least $\lfloor n/2 \rfloor +1$. Graphs that satisfy this property have been classified, and exist only for $n = 2,3,5,6$ \cite{Huber2017,Scott2004}, there exists however an infinite family of graphs whose minimal local sets are of size at least linear in their order, with constant 0.189 \cite{Javelle12}.

\smallskip 
\noindent{\bf Structure of the paper.} First, we define local sets and minimal local sets on a graph, and provide some examples, in \cref{sec:definitions}. Then we show how they can be defined alternatively using the cut-rank function, and prove additional results on the size and number of minimal local sets. Namely, we provide a tight bound on the size of minimal local sets, and show that a lower bound on the size of the minimal local sets implies a lower bound on their number. In \cref{sec:MLS_cover}, we prove the main result of this paper. For any graph, every vertex is contained in at least one minimal local set. We then give a polynomial-time algorithm that emerges from the proof of our main result, that finds a family of minimal local sets that cover all vertices of the graph. Finally, in \cref{sec:extension}, we extend the notion of local sets to $q$-multigraphs, the graphical counterpart of quantum qudit graph states, and show that our main result extends to any prime dimension. 

\section{Minimal local sets}
\label{sec:definitions}

This work focuses on the notion on minimal local sets. We give their definition and some basic properties in \cref{sec:basic_def}, then give an alternative definition using the cut-rank function in \cref{sec:alt_def}. We study the size and number of minimal local sets in \cref{subsec:additional}.

\subsection{Preliminaries}
\label{sec:basic_def}

Let us first give some notations and basic definitions. A graph $G$ is a pair $(V,E)$, where $V$ is the set of vertices, and $E\se V^2$ is the set of edges. Here we only consider graphs that are undirected (if $(u,v) \in E,~(v,u) \in E$) and simple ($\forall u \in V, (u,u) \notin E$). A cut is a  bipartition of the vertices of the graph; with a slight abuse of notation we use $A$ to denote the cut $\{A,V\setminus A\}$. The set $N_G(u) = \{v ~|~ (u,v) \in E\}$ is the neighborhood of $u$. For any $D \se V$, $Odd_G(D) = \Delta_{u\in D} N_G(u) = \{v \in V ~|~|N_G(v) \cap D| = 1 \text{ mod } 2\}$ is the odd neighborhood of $D$, where $\Delta$ denotes the symmetric difference on vertices. Informally, $Odd_G(D)$ is the set of vertices that are the neighbors of an odd number of vertices in $D$. A local complementation according to a given vertex $u$ consists in complementing the subgraph induced by the neighborhood of $u$, leading to the graph   $G\star u= G\Delta K_{N_G(u)}$ where $\Delta$ denotes the symmetric difference on edges and $K_A$ is the complete graph on the vertices of $A$. Two graphs are said locally equivalent  if there are related by a sequence of local complementations. 

\begin{definition}[Local set]Given $G=(V,E)$, a \emph{local set} $L$ is a non-empty subset of $V$ of the form $L = D \cup Odd_G(D)$ for some $D \se V$ called a \emph{generator}.
\end{definition}

Local sets are invariant under local complementation \!-\! hence their name:\! if $L$ is a local set in a graph, so is in any locally equivalent graph,\! but possibly with a distinct generator\! \cite{Perdrix06}.

\begin{definition}[Minimal local set]
    A minimal local set is a local set that is minimal by inclusion. 
\end{definition}

Minimal local sets satisfy the following property: under local complementation, any vertex of a minimal local set can be made its generator.

\begin{proposition}[\cite{Perdrix06}]
    For any minimal local set $L$ defined on a graph $G$, for any $x\in L$, there exists $G'$ locally equivalent  to $G$ such that $L = \{x\} \cup N_{G'}(x)$. 
\end{proposition}

We are interested in families of minimal local sets such that each vertex is contained in at least one minimal local set. We call such a family an \emph{MLS cover}.

\begin{definition}[MLS cover]
    Given a graph $G=(V,E)$, $\mathcal L\se 2^V$ is an MLS cover if
    \begin{itemize}[topsep=0.2em,itemsep=0.2em]
        \item $\forall L\in \mathcal L$, $L$ is a minimal local set of $G$,
        \item $\forall u\in V$, $\exists L \in \mathcal L$ such that $u\in L$.
    \end{itemize}
\end{definition}

Local sets, minimal local sets, and MLS covers are illustrated in \cref{fig:MLS}.

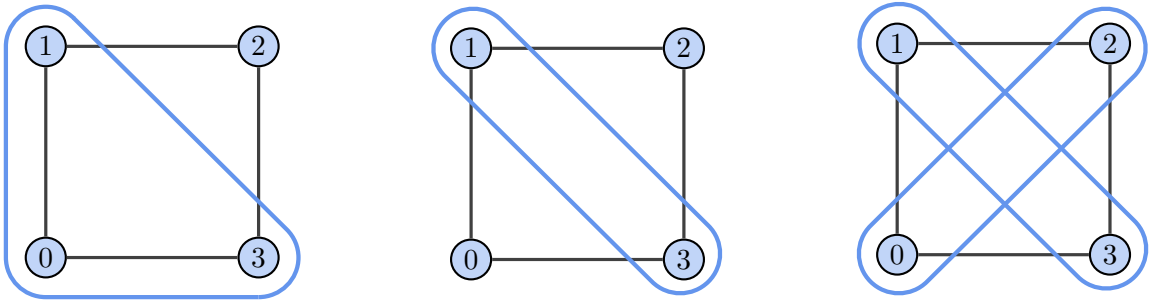
\begin{figure}[h]
    \centering
    
    \scalebox{1}{
    \begin{tikzpicture}[scale = 0.7]    
        \begin{scope}[every node/.style={circle,minimum size=15pt,thick,draw,fill=cornflowerblue!40, inner sep = 0pt}]
            \node (U0) at (0,0) {$0$};
            \node (U1) at (0,4) {$1$};
            \node (U2) at (4,4) {$2$};
            \node (U3) at (4,0) {$3$};
        \end{scope}
        \begin{scope}[every node/.style={},
                        every edge/.style={draw=darkgray,very thick}]              
            \path [-] (U0) edge node {} (U1);
            \path [-] (U1) edge node {} (U2);
            \path [-] (U2) edge node {} (U3);
            \path [-] (U3) edge node {} (U0);
        \end{scope}
        \draw[cornflowerblue, ultra thick] (4,-0.75) -- (0,-0.75) arc(-90:-180:0.75)-- (-0.75,0) -- (-0.75,4) arc(180:45:0.75);
        \draw[cornflowerblue, ultra thick] (4,-0.75)  arc(-90:45:0.75) -- (0.52,4.54);
    \end{tikzpicture}\qquad\qquad\raisebox{0cm}{
        \begin{tikzpicture}[scale = 0.7]
        
        \begin{scope}[every node/.style={circle,minimum size=15pt,thick,draw,fill=cornflowerblue!40, inner sep = 0pt}]
            \node (U0) at (0,0) {$0$};
            \node (U1) at (0,4) {$1$};
            \node (U2) at (4,4) {$2$};
            \node (U3) at (4,0) {$3$}; 
        \end{scope}
        
        \begin{scope}[every node/.style={},
                        every edge/.style={draw=darkgray,very thick}]              
            \path [-] (U0) edge node {} (U1);
            \path [-] (U1) edge node {} (U2);
            \path [-] (U2) edge node {} (U3);
            \path [-] (U3) edge node {} (U0);        
        \end{scope}
        \begin{scope}[shift={(3.4,-0.42)},rotate=45]
            \draw[cornflowerblue, ultra thick] (0,0) -- (0,5.5) arc(180:0:0.75) -- (1.5,0) arc(0:-180:0.75);
        \end{scope}
    \end{tikzpicture}}\qquad\qquad\raisebox{0cm}{
        \begin{tikzpicture}[scale = 0.7]
        
        \begin{scope}[every node/.style={circle,minimum size=15pt,thick,draw,fill=cornflowerblue!40, inner sep = 0pt}]
            \node (U0) at (0,0) {$0$};
            \node (U1) at (0,4) {$1$};
            \node (U2) at (4,4) {$2$};
            \node (U3) at (4,0) {$3$};  
        \end{scope}
        
        \begin{scope}[every node/.style={},
                        every edge/.style={draw=darkgray,very thick}]              
            \path [-] (U0) edge node {} (U1);
            \path [-] (U1) edge node {} (U2);
            \path [-] (U2) edge node {} (U3);
            \path [-] (U3) edge node {} (U0);        
        \end{scope}
        \begin{scope}[shift={(3.4,-0.42)},rotate=45]
            \draw[cornflowerblue, ultra thick] (0,0) -- (0,5.5) arc(180:0:0.75) -- (1.5,0) arc(0:-180:0.75);
        \end{scope}
        \begin{scope}[shift={(-0.5,0.55)},rotate=-45]
            \draw[cornflowerblue, ultra thick] (0,0) -- (0,5.5) arc(180:0:0.75) -- (1.5,0) arc(0:-180:0.75);
        \end{scope}
    \end{tikzpicture}}
    }
    
    \caption{(Left) A local set generated by $D = \{0\}$: $Odd_G(D) = \{1,3\}$. This is not a minimal local set. (Middle) A local set generated by $D = \{1,3\}$: $Odd_G(D) = \emptyset$. In particular, this is a minimal local set, as neither $\{1\}$, $\{3\}$ nor $\emptyset$ is a local set. (Right) an MLS cover of the graph, i.e. a set of minimal local sets such that each vertex is contained in at least one of them.}
    \label{fig:MLS}
    \end{figure}

The main goal of this paper is to prove that any graph has an MLS Cover, i.e. any vertex is contained in at least one minimal local set, and to provide an efficient algorithm to compute an MLS cover, given any graph. This is done in \cref{sec:MLS_cover}. First, for this purpose, we exhibit the links between (minimal) local sets and the cut-rank function.

\subsection{Links with the cut-rank function}
\label{sec:alt_def}

Given a graph $G$ and a cut $A$, one can define the map $\lambda_A : 2^{A}\to 2^{V\setminus A} = D\mapsto Odd_G(D)\setminus A$, which is linear with respect to the symmetric difference\footnote{$\forall D,D'\subseteq A, \lambda_G(D\Delta D') = \lambda_A(D)\Delta\lambda_A(D')$.}. Notice that $D$ is in the kernel of $\lambda_A$, i.e. $\lambda_A(D)=\emptyset$, if and only if $Odd_G(D)\subseteq A$, that is $D$ is a generator of a local set included in $A$. The rank of $\lambda_A$ is nothing but the so-called cut-rank of $A$. The cut-rank function, introduced by Bouchet under the name "connectivity function" \cite{Bouchet1993,Bouchet1987}, and coined "cut-rank" by Oum and Seymour \cite{OUM2006},  is usually defined as follows:

\begin{definition}[Cut-rank function]
    For $A \se V$, let the cut-matrix $\Gamma_A = ((\Gamma_A)_{ab}: a \in A\text{, } b \in V\sm A)$ be the matrix with coefficients in $\mathbb{F}_2$ (the finite field of order 2) such that $\Gamma_{ab} = 1$ if and only if $(a, b) \in E$. The cut-rank function of $G$ is defined as
    \begin{align*}
        \cutrk\colon 2^V & \longrightarrow \mathbb{N}\\
        A &\longmapsto \textbf{rank}_{\mathbb{F}_2}(\Gamma_A)
    \end{align*}
\end{definition}

A set $A$ is said \emph{full cut-rank} when $\cutrk(A)=|A|$. Minimal local sets can alternatively be defined using uniquely the cut-rank function. 

\begin{proposition}
    \label{prop:characMLS}Given a graph $G=(V,E)$ and $A\se V$, 
    \begin{itemize}[topsep=0.2em,itemsep=0.2em]
\item  If $A$ is a local set, then\footnote{This is not an equivalence. For example, in the complete graph of order 4 $K_4$, any set of vertices of size 3 is not a local set, however, such a set has a cut-rank of 1, and every set of size 1 or 2 also has a cut-rank of 1.} $\forall a \in A, \cutrk(A) \ls \cutrk(A \setminus \{a\})$,
\item     $A$ is a minimal local set if and only if $A$ is not full cut-rank, but each of its proper subset is, i.e.~$\forall a \in A, \cutrk(A) \ls \cutrk(A \setminus \{a\}) = |A|-1$.  
\end{itemize}
\end{proposition}

\begin{proof} Notice that $|A| - \cutrk(A)$, the dimension of the kernel of $\lambda_A$, counts the number of sets that generate a local set in $A$: $2^{|A|-\cutrk(A)} = |\{D \se A ~|~ D \cup Odd_G(D) \se A\}|$. As a consequence, $A \se V$ being a local set means that for any $a \in A$, there is more local sets (counted with their generators) in $A$ than in $A \setminus \{a\}$. Then $|A\sm \{a\}| - \cutrk(A\sm \{a\}) < |A| - \cutrk(A)$. This translates to $\cutrk(A) \ls \cutrk(A \setminus \{a\})$.

Conversely, if $\forall a \in A, \cutrk(A) \ls \cutrk(A \setminus \{a\}) = |A|-1$, then any proper subset of $A$ contains no local set, and $A$ is a local set, thus $A$ is a minimal local set.
\end{proof}

Two graphs that have the same cut-rank function have the same local sets\footnote{As local complementation does not change the rank of cut-matrices, this provides an alternative proof of the invariance of local sets under local complementation.}. This is an equivalence with the additional condition that the local sets have the same number of generators:

\begin{proposition}
    \label{prop:corres_mls_cutrankalt}
    Two graphs have the same cut-rank function if and only if they have the same local sets with the same number of generators.   
\end{proposition}

\begin{proof}
    The cut-rank function can be computed from the number of generators of each local set: for every $A \se V$, $$\cutrk(A) = |A| - log_{2}\left( 1 + \sum_{\text{$L$ local set in $A$}} |\{D \se L ~|~ D \cup Odd_G(D) = L\}|\right)$$
    Conversely, the number of generators of each local set  can be recursively computed from the values of the cut-rank function: $$ |\{D \se A | D \cup Odd_G(D) = A\}| = 2^{|A| - \cutrk(A)} - \sum_{B \varsubsetneq A} |\{D \se B | D \cup Odd_G(D) = B\}|$$ 
\end{proof}

The proposition below lists essential properties of the cut-rank function.

\begin{proposition}[\cite{OUM2006}]
    \label{prop:cutrank_prop}
    The cut-rank function satisfies the following properties:
    \begin{itemize}[topsep=0.2em,itemsep=0.2em]
        \item \textbf{symmetry: } $\forall A \se V,~\cutrk(V\sm A) =  \cutrk(A)$,
        \item \textbf{linear boundedness: } $\forall A \se V,~\cutrk(A)  \ls |A|$,
        \item \textbf{submodularity: } $\forall A,B \se V,~\cutrk(A \cup B) + \cutrk(A \cap B)  \ls \cutrk(A) + \cutrk(B)$.
    \end{itemize}   
\end{proposition}

The conjunction of symmetry, linear boundedness and submodularity, together with the fact that $\cutrk$ has values in $\mathbb{N}$, implies the following useful properties:

\begin{proposition} \label{prop:cutrank_prop2}
    For any graph $G=(V,E)$ of order $n$, 
    \begin{itemize}[topsep=0.2em,itemsep=0.2em]
        \item[$(i)$] $\cutrk(\emptyset)= \cutrk(V) = 0$,
        \item[$(ii)$]  $\forall A \se V,\cutrk(A) \ls \lfloor n/2 \rfloor$, so a full-cut-rank set is of size at most $\lfloor n/2 \rfloor$,
        \item[$(iii)$]  If $\cutrk(K)=|K|$ then for any $A\se K$, $\cutrk(A)=|A|$, i.e. any subset of a full-cut-rank set is full cut-rank.    
    \end{itemize} 
\end{proposition}

\begin{proof}
    $(i)$ Using {linear boundedness}, $\cutrk(\emptyset) \ls |\emptyset| = 0$ so $\cutrk(\emptyset) = 0$ as $\cutrk$ has values in $\NN$. Then, using {symmetry}, $\cutrk(V) = \cutrk(\emptyset) = 0$. $(ii)$
  Let $A \se V$. Using {linear boundedness}, $\cutrk(A) \ls |A|$. Also, using {symmetry}, $\cutrk(A) = \cutrk(V\sm A) \ls |V\sm A| = n - |A|$. Thus, $\cutrk(A) \ls \min(|A|, n - |A|) \ls \lfloor n/2 \rfloor$.
    $(iii)$ Let $A\se K$. Using {submodularity}, $\cutrk(A)+\cutrk(K\sm A)\gs \cutrk(K)+\cutrk(\emptyset)=|K|+0$. Moreover, using {linear boundedness}, $\cutrk(A)\ls |A|$ and $\cutrk(K\sm A)\ls |K|-|A|$, so  $\cutrk(A)= |A|$.
\end{proof}

\subsection{Size and number of minimal local  sets}
\label{subsec:additional}

According to \cref{prop:cutrank_prop2} and \cref{prop:characMLS}, minimal local sets are of size at most half the order of the graph. We slightly refine this bound and show it is tight:

\begin{restatable}{proposition}{MLSbounds}
    For any minimal local set $A$ in a graph of order $n$,
     \[|A|\le \begin{cases}   \hspace{0.15cm}n/2&\text{if $n= 0 \bmod 4$}\\   \lfloor n/2 \rfloor + 1 &\text{otherwise}\end{cases}\]
    This bound is tight in the sense that for every $n>0$ there exists a graph of order $n$ that contains a minimal local set of this particular size.
    \label{prop:MLSbounds}
\end{restatable}

\begin{proof}
    If $n \neq 0 \bmod 4$, the bound derives from \cref{prop:characMLS} along with \cref{prop:cutrank_prop2}. Otherwise, the bound is slightly stronger due to constraints of the cut-rank function. See details, along with explicit constructions of graphs that attain the bound, in \cref{app:proof_MLSbounds}.
\end{proof}

Small and large minimal local sets can coexist in a graph, this is for instance the case for a path graph $P_n$ of order $n>2$ which has minimal local sets of any size ranging from $2$ to $\lceil n/2 \rceil$ (see   \cref{sec:MLSPath}). Complete graphs, and bipartite complete graphs only have small minimal local sets (which are respectively any pair of vertices and any pair of vertices in the same partition). There exist also graphs with only `large' minimal local sets, namely at least $0.189n$ where $n$ is the order of the graph. Indeed, the size of the smallest local set is nothing but the local minimum degree of the graph plus one \cite{Perdrix06}, and it has been proved that there exist graphs with such a large local minimum degree \cite{Javelle12}.

The size of the local sets is also, to some extent, related to the number of minimal local sets in a graph. For instance, complete graphs and bipartite complete graphs have a number of minimal local sets quadratic in their order. More generally, if the size of the minimal local sets is upper bounded by $k$ then there are obviously $O(n^k)$ minimal local sets in a graph of order $n$. Maybe more surprisingly, a lower bound on the size of the minimal local sets implies a lower bound on their number: 

\begin{restatable}{proposition}{expnumberMLS}    
Given a graph $G$ of order $n$, if all minimal local sets of $G$ are of size at least $m$, then the number of minimal local sets in $G$ is at least \[\frac{1-2r}{3\sqrt n}2^{n\left(1-(1-r)H_2\big(\frac1{2(1-r)}\big)\right)}\]
where $r=\frac m n$ and $H_2(x) = -x\log_2(x) - (1-x)\log_2(1-x)$ is the binary entropy. 
\label{prop:exp_number_MLS}
\end{restatable}

\begin{proof}
\cref{prop:cutrank_prop2} along with \cref{prop:characMLS} imply that any set with more than $n/2$ vertices contains a minimal local set. One can upper bound the number of sets of size $n/2$ in which a given fixed minimal local set is contained. Roughly speaking, a large minimal local set is contained in fewer sets of size $n/2$ than a small minimal local set. Thus, a counting argument implies the lower bound on the number of minimal local sets. See details in \cref{app:proof_exp_number_MLS}.
\end{proof}

\cref{prop:exp_number_MLS} implies that a graph of order $n$, in which the size of the minimal local sets is lower bounded by $cn$ for some constant $c>0$,  has an exponential number of minimal local sets. In particular, graphs of order $n$ and local minimum degree of size at least $0.189n$ have at least $1.165^n$ minimal local sets. Notice that there also exist explicit graphs with exponentially many minimal local sets, this is for instance the case of the graph $K_{k,k} \Delta M_k$ defined as the symmetric difference of a complete bipartite graph and a matching, which is of order $n=2k$ and has more than $2^{k-1}=\frac12\sqrt 2^n$ minimal local sets (see \cref{app:bipartite_matching}).

As the number of minimal local sets can be exponential in the order of a graph, we focus, in the next section,  on MLS covers, which are representative families of minimal local sets covering any vertex of a graph. 
Notice that if a graph has an MLS cover, then it admits an MLS cover made of at most a linear number of minimal local sets.

\section{MLS cover}
\label{sec:MLS_cover}

To cover a vertex $u$ of a graph $G$ with a minimal local set, one can consider the local set generated by $u$, i.e. its closed neighborhood $N_G[u]=\{u\}\cup N_G(u)$. However, such a  local set is not always minimal, worse yet, it does not necessarily contain a minimal local set that includes $u$. Indeed, in the cycle $C_4$ of order 4, for every vertex $u$, $N_G[u]$ contains a single minimal local set that does not include $u$ (see \cref{fig:MLS}). 

In this particular $C_4$ example however, one can notice that, roughly speaking, minimizing local sets of the form $N_G[u]$ is enough to get an MLS cover: considering the minimal local sets included in closed neighborhoods is enough to produce an MLS cover.

This strategy does not work in general. There exist graphs (see \cref{fig:MLSCover_is_hard}) for which an MLS cover requires minimal local sets that are not subsets of any closed neighborhood. Thus, one can wonder what is the best strategy to produce an MLS cover and even whether it exists for every graph. We show in this section that every graph admits an MLS cover (\cref{thm:MLS_cover}, proved in  \cref{sec:MLS_cover_proof}) and then introduce an efficient algorithm for finding an MLS cover (\cref{sec:algorithm}).

\begin{figure}[h]
    \centering
    \scalebox{0.8}{
    \begin{tikzpicture}[xscale = 0.9, yscale =1]    
        \begin{scope}[every node/.style={circle,minimum size=15pt,thick,draw,fill=cornflowerblue!40, inner sep = 0pt}]
            \node (U1) at (0,0) {$1$};
            \node (U0) at (-10,0) {$0$};
            \node (U2) at (10,0) {$2$};
            \node (U3) at (-5,3.75) {$3$};
            \node (U4) at (-5,2.25) {$4$};
            \node (U5) at (-5,0.75) {$5$};
            \node (U6) at (-5,-0.75) {$6$};
            \node (U7) at (-5,-2.25) {$7$};
            \node (U8) at (-5,-3.75) {$8$};
            \node (U9) at (5,3.75) {$9$};
            \node (U10) at (5,2.25) {$10$};
            \node (U11) at (5,0.75) {$11$};
            \node (U12) at (5,-0.75) {$12$};
            \node (U13) at (5,-2.25) {$13$};
            \node (U14) at (5,-3.75) {$14$};
        \end{scope}
        \begin{scope}[every node/.style={},
                        every edge/.style={draw=darkgray,very thick}]              
            \path [-] (U0) edge node {} (U3);
            \path [-] (U0) edge node {} (U4);
            \path [-] (U0) edge node {} (U5);
            \path [-] (U0) edge node {} (U6);
            \path [-] (U0) edge node {} (U7);
            \path [-] (U0) edge node {} (U8);
            \path [-] (U1) edge node {} (U3);
            \path [-] (U1) edge node {} (U4);
            \path [-] (U1) edge node {} (U5);
            \path [-] (U1) edge node {} (U6);
            \path [-] (U1) edge node {} (U7);
            \path [-] (U1) edge node {} (U8);
            \path [-] (U1) edge node {} (U9);
            \path [-] (U1) edge node {} (U10);
            \path [-] (U1) edge node {} (U11);
            \path [-] (U1) edge node {} (U12);
            \path [-] (U1) edge node {} (U13);
            \path [-] (U1) edge node {} (U14);
            \path [-] (U2) edge node {} (U9);
            \path [-] (U2) edge node {} (U10);
            \path [-] (U2) edge node {} (U11);
            \path [-] (U2) edge node {} (U12);
            \path [-] (U2) edge node {} (U13);
            \path [-] (U2) edge node {} (U14);
            \path [-] (U3) edge node {} (U4);
            \path [-] (U5) edge node {} (U6);
            \path [-] (U7) edge node {} (U8);
            \path [-] (U9) edge node {} (U10);
            \path [-] (U11) edge node {} (U12);
            \path [-] (U13) edge node {} (U14);

        \end{scope}
    \end{tikzpicture}}
    
    \caption{Example of a graph $G$ of order 15 where a naive approach for finding an MLS cover, based on the neighborhood of every vertex, cannot work. The vertices 0, 1 and 2 are contained in several minimal local sets, the most obvious one being $\{0,1,2\}$. However, no local set of the form $N_G[u]$ contains a minimal local set that contains either 0, 1 or 2. Indeed, the minimal local sets in $N_G[0]$ are $\{3,4\}$, $\{5,6\}$ and $\{7,8\}$. Symmetrically, the minimal local sets in $N_G[2]$ are $\{9,10\}$, $\{11,12\}$ and $\{13,14\}$. The minimal local sets in $N_G[1]$ are $\{3,4\}$, $\{5,6\}$, $\{7,8\}$, $\{9,10\}$, $\{11,12\}$ and $\{13,14\}$. The only minimal local set in $N_G[3]$ is $\{3,4\}$, and the same goes for 4, 5, 6, 7, 8, 9, 10, 11, 12, 13 and 14.}
    \label{fig:MLSCover_is_hard}
\end{figure}
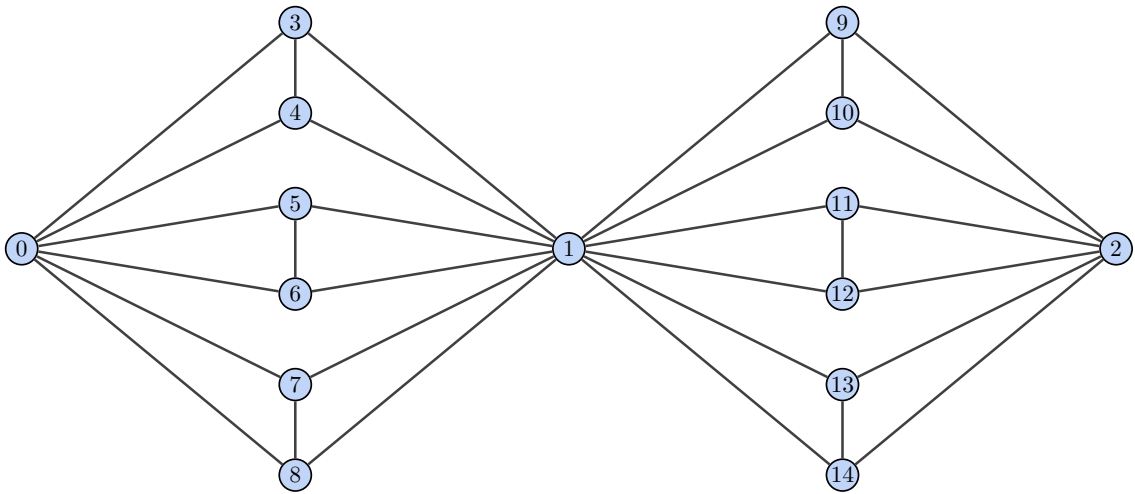

\begin{theorem}
    \label{thm:MLS_cover}
    Any graph has an MLS cover.   
\end{theorem}

\subsection{Proof of Theorem \ref{thm:MLS_cover}}
\label{sec:MLS_cover_proof}

In this subsection, we prove \cref{thm:MLS_cover}. First, we introduce a cut-rank-based characterisation of the existence of a minimal local set covering a given vertex:

\begin{lemma}
    \label{lemma:characinMLS} Given a graph $G=(V,E)$, 
    $a\in V$ is contained in a minimal local set if and only if there exists $A\se V$ such that $A$ is full cut-rank but $A\cup \{a\}$ is not.
\end{lemma}

\begin{proof}
Let $a \in V$. Suppose $a$ is contained in a minimal local set $A \se V$. Then $A \sm \{a\}$ is full cut-rank but $A$ is not.
Conversely, suppose that there exists $A\se V$ such that $A$ is full cut-rank but $A\cup \{a\}$ is not (obviously, $a \notin A$). According to \cref{prop:cutrank_prop2}, any set $B \se A$ is full cut-rank. Among all sets $B \se A$ such that $B \cup \{a\}$ is not full cut-rank (there is at least one: $A$), take one of them that is minimal by inclusion. Call it $C$. Let us show that $C \cup \{a\}$ is a minimal local set. By hypothesis, $C \cup \{a\}$ is not full cut-rank. Every subset of $C$ is full cut-rank, because such a set is also a subset of $A$. Also, every subset of $C \cup \{a\}$ containing $a$ is full cut-rank, by minimality of $C$. So, according to \cref{prop:characMLS},  $C \cup \{a\}$ is a minimal local set. So $a$ is contained in some minimal local set.
\end{proof}

Notice that a set is full cut-rank when the cut-rank function  is locally strictly increasing in the following sense:

\begin{lemma}
    \label{lemma:localdec}Given a graph $G=(V,E)$, 
   $C \se V$ is full cut-rank if and only if $\forall a\in C$, $\cutrk(C\sm \{a\}) < \cutrk(C)$.
\end{lemma}
\begin{proof}
$(\Rightarrow)$ follows from \cref{prop:cutrank_prop2}.
The proof of $(\Leftarrow)$ is by induction on the size of $C$. The property is obviously true when $|C|\ls 1$. 
Assume $|C|\gs 2$, for any $b\in C$ and any $a\in C\sm \{b\}$, using {submodularity},
\begin{equation*}
    \begin{split}
        \cutrk((C\sm \{a\}) \cup (C\sm \{b\}))\\ +~\cutrk((C\sm \{a\}) \cap (C\sm \{b\})) & \ls \cutrk(C\sm \{a\}) + \cutrk(C\sm \{b\}) \\
        \cutrk(C) + \cutrk(C\sm \{a,b\}) & <  \cutrk(C) + \cutrk(C\sm \{b\})\\
        \cutrk((C\sm \{b\})\sm\{a\}) & <  \cutrk(C\sm \{b\})  
    \end{split}
\end{equation*}
 By induction hypothesis, $C\sm \{b\}$ is full cut-rank, so $C$ is also full cut-rank as $\cutrk(C)>\cutrk(C\sm \{b\})=|C\sm \{b\}| = |C|-1$.
\end{proof}

When a set is not full cut-rank, one can find a sequence of nested subsets with a larger cut-rank, leading to a full cut-rank subset:

\begin{lemma} \label{lemma:exists_full}
    Given a graph $G=(V,E)$,  for any $A\se V$, there exist $B_0\varsubsetneq \ldots \varsubsetneq B_{|A|-\cutrk(A)-1} \varsubsetneq B_{|A|-\cutrk(A)} = A$ such that $\forall i$, $|B_i| = \cutrk(A)+i$ and $\cutrk(B_i)\gs \cutrk(A)$. In particular, $B_0$ is full cut-rank.
\end{lemma}

\begin{proof}
    Note this makes sense because of the {linear boundedness} of $\cutrk$, which ensures that $|A|-\cutrk(A) \gs 0$. It is enough to show that for any $B\se A$, if $\cutrk(B)\gs \cutrk(A)$ and $|B| >\cutrk(A)$ then  $\exists a\in B$ such that $\cutrk(B\sm \{a\})\gs \cutrk(A)$. There are two cases: (i) if $B$ is full cut-rank then according to \cref{prop:cutrank_prop2} for any $a\in B$, $B\sm\{a\}$ is also full cut-rank so $\cutrk(B\sm\{a\}) = |B|-1\gs \cutrk(A)$; (ii) if $B$ is not full cut-rank then according to \cref{lemma:localdec}, $\exists a\in B$ such that $\cutrk(B\sm \{a\}) \gs \cutrk(B)$, so $\cutrk(B\sm \{a\}) \gs \cutrk(A)$.
\end{proof}

An interesting consequence of \cref{lemma:exists_full} is that for any full-cut-rank set, there exists a disjoint full-cut-rank set of same cardinality.

\begin{corollary}
    \label{cor:full_avoid}Given a graph $G=(V,E)$, 
    for any full-cut-rank set $A \se V$, $\exists B \in {V \sm A \choose |A|}$ that is full cut-rank\footnote{Given a set $K$ and an integer $k$, ${K \choose k}$ refers to $\{B \se K ~|~|B|=k\}$.}.
\end{corollary}

\begin{proof}
Let $A \subseteq V$ be a full-cut-rank set, i.e. $\cutrk(A) = |A|$. Using {symmetry}, $\cutrk(V \sm A) = |A|$. Then applying \cref{lemma:exists_full} on $V \sm A$ yields the result.
\end{proof}

We are now ready to conclude our reasoning and prove that any vertex $a$ of an arbitrary graph $G=(V,E)$ is contained in some minimal local set.
Let $r$ be the maximal cardinality of a full-cut-rank set in $V$, i.e. $r=\max\{|A|~|~\cutrk(A)=|A|\}$, and let $A\se V$ be a full-cut-rank set of size $r$. For any vertex $a\in V$, 
\begin{itemize}[topsep=0.2em,itemsep=0.2em]
\item If $a\notin A$ then by maximality of $r$, $A\cup\{a\}$ is not full cut-rank, thus by \cref{lemma:characinMLS}, $a$ is contained in a minimal local set.
\item If $a\in A$, then according to \cref{cor:full_avoid}, there exists a full-cut-rank set $B\in {V \sm A \choose r}$. By maximality of $r$, $B\cup\{a\}$ is not full cut-rank, so, according to \cref{lemma:characinMLS}, $a$ is contained in a minimal local set. 
\end{itemize}

This concludes the proof.

\subsection{A polynomial-time algorithm for  finding an MLS cover}

\label{sec:algorithm}

In this subsection, we turn \cref{thm:MLS_cover} into a polynomial-time algorithm that generates an MLS cover for any given input graph. We provide in the following a textual description of the algorithm, the corresponding pseudo-code is given in \cref{app:algorithm}.

\smallskip\noindent{\bf Description of the algorithm.} 
The core of the algorithm consists, given a graph $G$ and a vertex $a$, in producing a minimal local set that contains $a$. One can then iterate on uncovered vertices to produce an MLS cover of the graph. 

Given a vertex $a$, a minimal local set that contains $a$ is produced as follows, using essentially two stages:

($i$) First, the algorithm produces a full-cut-rank set $A$ such that $A \cup\{a\}$ is not full cut-rank. The procedure consists in starting with an empty set $A$ -- which is full cut-rank -- and then increasing the size of $A$, in a full-cut-rank preserving manner, until $A\cup \{a\}$ is not full cut-rank. To increase the size of $A$, notice that  if $A\cup \{a\}$ is full cut-rank then,  
according to \cref{lemma:exists_full}, there exists a disjoint set $A'$ -- so in particular $a\notin A'$ -- such that $|A'|=|A|+1$ and $A'$ is full cut-rank. The proof of \cref{lemma:exists_full} is constructive: starting from $C= V\setminus (A\cup \{a\})$ some vertices are removed  one-by-one  from $C$ to get the set $A'$. To produce such a set $A'$ from $C$, there are $O(n)$ vertices to remove, at each step there are $O(n)$ candidates, and deciding whether a vertex can be removed  costs a constant number  of evaluations of the cut-rank function. 

($ii$) Given a full-cut-rank set $A$ such that $A\cup \{a\}$ is not full cut-rank, \cref{lemma:characinMLS} guarantees the existence of a minimal local set that contains $a$. Notice that the proof of \cref{lemma:characinMLS} consists in finding, among all  subsets $B$ of $A$ such that $B$ is full cut-rank but $B\cup \{a\}$ is not, one that is minimal by inclusion. To do so, one can start from $A$ and remove one-by-one vertices from $A$ until reaching a set that is minimal by inclusion. It takes $O(n)$ steps\footnote{We only need to consider each vertex at most once. Indeed, at each step, if $A \sm \{b\}$ is full cut-rank, for any $B \se A$, $B \sm \{b\}$ is also full cut-rank by \cref{prop:cutrank_prop2}.}, each step involving a constant number of evaluations of the cut-rank function. 

\smallskip\noindent{\bf Complexity.} Stage $(i)$ uses $O(n^3)$ evaluations of the cut-rank function, and stage $(ii)$ only $O(n)$, so the overall algorithm uses $O(n^4)$ evaluations of the cut-rank function. Using Gaussian elimination, the cut-rank can be computed in $O(n^\omega)$ field operations \cite{Bunch1974,IBARRA198245} where $\omega < 2.38$.

\subsection{Generalization to any symmetric, linearly bounded and submodular  function with values in $\mathbb{N}$} 
\label{subsec:MLSCover_generalization}

It has to be noted that the existence of an MLS cover and the fact that one can be computed efficiently remains true for any notion equivalent to minimal local sets defined with a function that shares some properties of the cut-rank function. Precisely, given a set $V$, and a function $\mu:2^V\to \NN$ that satisfies symmetry, linear boundedness, and submodularity, define a $\mu$-minimal local set $A\se V$ as a set that satisfies $\mu(A) \ls |A| - 1$ and $\forall B \varsubsetneq A$, $\mu(B) = |B|$. Then $V$ is covered by its $\mu$-minimal local sets, and a "$\mu$-MLS cover" can be computed efficiently (given that the values of $\mu$ can be computed efficiently). Indeed, every proof we give in \cref{sec:MLS_cover} (as the proof of \cref{prop:cutrank_prop2}) uses only the fact that $\cutrk:2^V\to \NN$ satisfies symmetry, linear boundedness, and submodularity.

An example of a function satisfying such properties is the connectivity function of a matroid. Given a matroid $(E,r)$ where $E$ is the ground set and $r$ is the rank function, the connectivity function is the function $\lambda$ that maps any set $X \subseteq E$ to $\lambda(X) = r(X)+r(E\sm X)-r(E)$.



\section{Extension to $q$-multigraphs}

\label{sec:extension}

While (simple, undirected) graphs correspond to quantum qubit graphs states, their natural higher dimension extension are multigraphs, which correspond to quantum qudit graph states \cite{Beigi2006,Ketkar2006}. A quick introduction to multigraphs and qudit graph states can be found in \cite{marin2013}. Here we explain how our main result extends to $q$-multigraphs.

\begin{definition}[$q$-multigraphs]
    Given a prime number $q$, a $q$-multigraph $G$ is a pair $(V, \Gamma)$ where $V$ is the set of vertices and $\Gamma : V \times V \longrightarrow \mathbb{F}_q$ is the adjacency matrix of $G$: for any $u,v \in G$, $\Gamma(u,v)$ is the multiplicity of the edge $(u,v)$ in $G$.
\end{definition}

Here we consider undirected simple $q$-multigraphs, i.e. $\forall u,v \in V$, $\Gamma(u,v) = \Gamma(v,u)$ and $\Gamma(u,u)=0$.
The cut-rank function on $q$-multigraphs is defined similar to graphs (which correspond to $2$-multigraphs): 

\begin{definition}[Cut-rank function on $q$-multigraphs]
    Let $G=(V, \Gamma)$ be a $q$-multigraph. For $A \se V$, let the cut-matrix $\Gamma_A = ((\Gamma_A)_{ab}: a \in A\text{, } b \in V\sm A)$ be the matrix with coefficients in $\mathbb{F}_q$ such that $\Gamma_{ab} = \Gamma(a,b)$. The cut-rank function of $G$ is defined as
    \begin{align*}
        \cutrk\colon 2^V & \longrightarrow \mathbb{N}\\
        A &\longmapsto \textbf{rank}_{\mathbb{F}_q}(\Gamma_A)
    \end{align*}
\end{definition}

It is not hard to see that the cut-rank function on $q$-multigraphs satisfies the properties of \cref{prop:cutrank_prop}, namely {symmetry}, {linear boundedness}, and {submodularity}. As explained in \cref{subsec:MLSCover_generalization}, \cref{thm:MLS_cover} extends to any function with values in $\mathbb{N}$ that satisfies these properties. To extend minimal local sets to $q$-multigraphs, the easier is to use their characterisation in terms of the cut-rank function (see \cref{prop:characMLS}).

\begin{definition}[Minimal local sets and MLS covers on $q$-multigraphs]
    Given $G=(V, \Gamma)$ a $q$-multigraph,
    \begin{itemize}[topsep=0.2em,itemsep=0.2em]
 \item    $A \se V$ is a \emph{minimal local} set if $A$ is a non-full-cut-rank set whose proper subsets are all full cut-rank, i.e. $\forall a \in A, \cutrk(A) \ls \cutrk(A \sm a) = |A|-1$,
 \item    $\mathcal L\se 2^V$ is an \emph{MLS cover} if $\forall L\in \mathcal L$, $L$ is a minimal local set of $G$, and $\forall u\in V$, $\exists L \in \mathcal L$ such that $u\in L$.
 \end{itemize}
\end{definition}

Local complementation naturally extends to $q$-multigraphs (see \cite{marin2013} for a definition), and two locally equivalent $q$-multigraphs have the same cut-rank function \cite{Kante2007}. This implies that again, in the case of $q$-multigraphs, minimal local sets are invariant under local complementation. Our main result also extends to $q$-multigraphs.

\begin{theorem}
    Any q-multigraph has an MLS cover. 
\end{theorem}

\section{Conclusion}

In this work, we have shown that any graph admits an MLS cover, thus every vertex is contained in a minimal local set. We have also introduced a polynomial time algorithm for producing an MLS cover of a given graph. Our results rely on a cut-rank-based characterisation of minimal local sets, and the peculiar properties of the cut-rank function. As a consequence, our results extend to any set function that is symmetric, linearly bounded, and submodular; this is how we generalized our results to $q$-multigraphs for instance. We have also proved a few results on the size and the number of minimal local sets in a graph, which open a few directions for future work:
 \begin{itemize}[topsep=0.2em,itemsep=0.2em]
\item  One can define the minimal local set number of a graph as the smallest number  of minimal local sets needed to cover all vertices of a graph; the complexity of the associated decision problem is unknown, up to our knowledge. \item Given an MLS cover $\{L_i\}$ of a graph, one can define its intersection graph which vertices are the minimal locals sets $L_i$ and two vertices are connected if the corresponding minimal local sets have a non-empty intersection. What is the structure of such intersection graph? Notice for instance that the cycle of order $4$ has a single MLS cover which intersection graph is empty. Is there always an MLS cover which intersection graph is a forest? Can it be exploited to speed up some algorithms? Could it be related to the rank-width parameter?
\item As the number of minimal local sets of a graph can be exponentially large, one can wonder whether there exists a polynomial delay algorithm for enumerating all minimal local sets of a graph. 
\end{itemize}

\section*{Acknowlegments} 

The authors thank Mathilde Bouvel and Noé Delorme for fruitful discussions. This work is supported by the PEPR integrated project EPiQ ANR-22-PETQ-0007 part of Plan France 2030, by the STIC-AmSud project Qapla’ 21-STIC-10, and by the European projects NEASQC and HPCQS.

\bibliographystyle{plainurl}
\bibliography{ref}

\appendix

\section{Complete proof of Proposition \ref{prop:MLSbounds}}
\label{app:proof_MLSbounds}

\MLSbounds*

\begin{proof}

First note that any minimal local set has size at most $\lfloor n/2 \rfloor + 1$. Indeed, by \cref{prop:characMLS}, a necessary condition for a set $A \se V$ to be a minimal local set is that any proper subset of $A$ is full cut-rank. And, by \cref{prop:cutrank_prop2}, any full-cut-rank set has size at most $ \lfloor n/2 \rfloor$. This proves the result for $n \neq 0 \bmod 4$. Let us now suppose that $n = 0 \bmod 4$, and suppose by contradiction that there exists a minimal local set $A$ of size $n/2 + 1$. Using {symmetry} and {linear boundedness}, $\cutrk(A) = \cutrk(V \sm A) \ls |V \sm A| = n/2 - 1$. Then $|A| - \cutrk(A) \gs 2$. As $A$ is a minimal local set, $|A| - \cutrk(A) =$ 1 or 2, and the latter can only occur when $|A|$ is even: this result was proved in \cite{VandenNest05}, and we provide a graph-theoretical proof in \cref{app:1or3}. This leads to a contradiction, as $|A| = n/2 + 1$ is odd.\\
We show below that this bound is tight, in the sense that for any $n$ there exists a graph of order $n$ that has a minimal local set whose size matches the bound. For this purpose, we explicitly construct graphs of arbitrary order that contain a minimal local set whose size matches the bound.

{\bf Case $n$ odd.}\\ 
Let $n = 2 m + 1$ be an odd number. We provide in \cref{fig:bounds1} a construction of a graph $G$ that contains a minimal local set of size $m+1$.

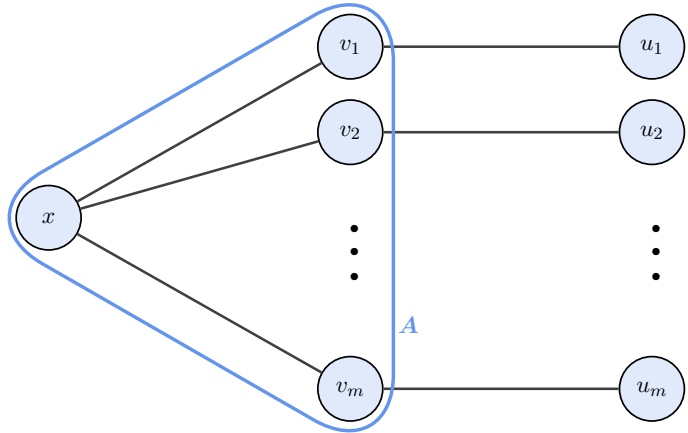
\begin{figure}[h!]
\centering
\scalebox{\valuescale}{
\begin{tikzpicture}[scale = 0.7]
\begin{scope}[every node/.style={circle,minimum size=30pt,thick,draw, fill=cornflowerblue!20}]
    \node (A) at (0,4) {$x$};
    \node (B) at (7,8) {$v_1$};
    \node (C) at (7,6) {$v_2$};
    \node (D) at (7,0) {\scalebox{1}{$v_{m}$}};
    \node (E) at (14,8) {$u_1$};
    \node (F) at (14,6) {$u_2$} ;
    \node (G) at (14,0) {\scalebox{1}{$u_{m}$}} ;
\end{scope}
\begin{scope}[every node/.style={},
                every edge/.style={draw=darkgray,very thick}]
    \path [-] (A) edge node {} (B);
    \path [-] (A) edge node {} (C);
    \path [-] (A) edge node {} (D);
    \path [-] (B) edge node {} (E);
    \path [-] (C) edge node {} (F);
    \path [-] (D) edge node {} (G);
\end{scope}
\begin{scope}[style={draw=blue, ultra thick}]
    \draw (7.1,3.2) node[rotate = 90](){\scalebox{2}{$\boldsymbol{\ldots}$}};
    \draw (14.1,3.2) node[rotate = 90](){\scalebox{2}{$\boldsymbol{\ldots}$}};
\end{scope}
\begin{scope}[style={draw=cornflowerblue, ultra thick}]
    \draw[rounded corners=15mm] (-2,4)--(8,9.8)--(8,-1.8)--cycle;
    \draw (8.35,1.5) node[text=cornflowerblue](){$\boldsymbol{A}$};
\end{scope}
\end{tikzpicture}
}
\caption{Illustration of a minimal local set of size $\lceil n/2 \rceil$ in $G$.}
\label{fig:bounds1}
\end{figure}

$A$ is a local set, as $A = D \cup Odd_G(D)$ with $D = \{x\}$.
In particular, $A$ is a minimal local set. Indeed, let us assume that there exists a subset $D'$ such that $D' \cup Odd_G(D') \varsubsetneq A$. It cannot contain any of the $v_i$'s as we would then also have $u_i \in Odd_G(D')$. So $D' = \emptyset$ or $\{x\}$, leading to a contradiction.

{\bf Case $n$ even and $n/2$ odd.}\\  
Let $n = 2 m$ be an even number such that $m$ is odd. We provide in \cref{fig:bounds2} a construction of a graph $G$ that contains a minimal local set of size $m+1$.
\begin{figure}[h!]
\centering
\scalebox{\valuescale}{
\begin{tikzpicture}[scale = 0.7]
\begin{scope}[every node/.style={circle,minimum size=30pt,thick,draw,fill=cornflowerblue!20}]
    \node (X) at (0,10) {$x$};
    \node (U) at (7,10) {$y$};

    \node (B) at (0,8) {$v_1$};
    \node (C) at (0,6) {$v_2$};
    \node (D) at (0,0) {\scalebox{0.8}{$v_{m-1}$}};

    \node (E) at (7,8) {$u_1$};
    \node (F) at (7,6) {$u_2$} ;
    \node (G) at (7,0) {\scalebox{0.8}{$u_{m-1}$}} ;
\end{scope}
\begin{scope}[every node/.style={},
                every edge/.style={draw=darkgray,very thick}]
    \path [-] (X) edge node {} (U);
    \path [-] (B) edge node {} (E);
    \path [-] (C) edge node {} (F);
    \path [-] (D) edge node {} (G);

    \path [-] (X) edge node {} (B);
    \path [-] (U) edge node {} (E);
    \path [-] (B) edge node {} (C);
    \path [-] (E) edge node {} (F);

    \path [-] (X) edge[bend right=50] node {} (C);
    \path [-] (X) edge[bend right=50] node {} (D);

    \path [-] (U) edge[bend left=50] node {} (F);
    \path [-] (U) edge[bend left=50] node {} (G);

    \path (C) edge node {} (0,4.8);
    \path (F) edge node {} (7,4.8);
    \path (D) edge node {} (0,1.2);
    \path (G) edge node {} (7,1.2);

\end{scope}
\begin{scope}[style={draw=blue, ultra thick}]
    \draw (0,3.2) node[rotate = 90](){\scalebox{2}{$\boldsymbol{\ldots}$}};
    \draw (7,3.2) node[rotate = 90](){\scalebox{2}{$\boldsymbol{\ldots}$}};
\end{scope}
\begin{scope}[style={draw=cornflowerblue, ultra thick}]
    \draw (7,9) arc(-90:90:1) -- (-0,11) arc(90:180:1) -- (-1,0) arc(-180:-90:1) -- (0,-1) arc(-90:0:1) -- (1,8) arc(180:90:1) -- (7,9);
    \draw (1.3,1.8) node[text=cornflowerblue](){$\boldsymbol{A}$};
\end{scope}
\end{tikzpicture}
}
\caption{Illustration of a minimal local set of size $ n/2+1$ in $G$.}
\label{fig:bounds2}
\end{figure}
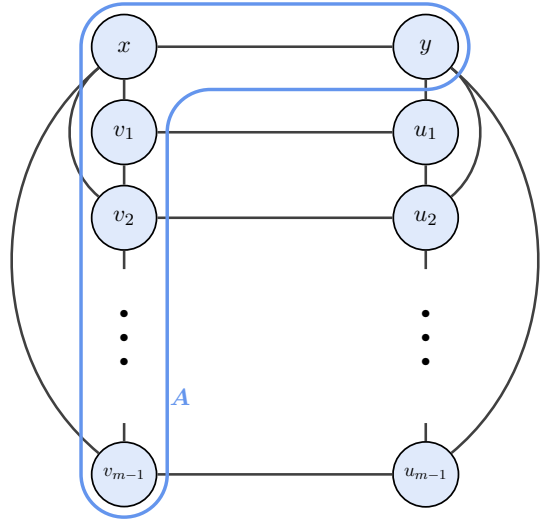

$A$ is a local set, as $A = D \cup Odd_G(D)$ with $D = \{x\}$. 
In particular, $A$ is a minimal local set. Indeed, let us assume that there exists a subset $D'$ such that $D' \cup Odd_G(D') \varsubsetneq A$. If it contains $y$, it also needs to contain every $v_i$'s, so that none of the $u_i$'s enter $Odd_G(D')$. At this point, either $x \in D'$, or $x \notin D'$: in both cases, $D' \cup Odd_G(D') = A$ (because $m$ is odd), leading to a contradiction. If it does not contain $y$, it cannot contain any of the $v_i$'s, so that none of the $u_i$'s enter $Odd_G(D')$. So $D' = \{x\} = D$, leading to a contradiction.

{\bf Case $n$ even and $n/2$ even.}\\ 
Let $n = 2 m$ be an even number such that $m$ is even. We provide in \cref{fig:bounds3} a construction of a graph $G$ that contains a minimal local set of size $m$.

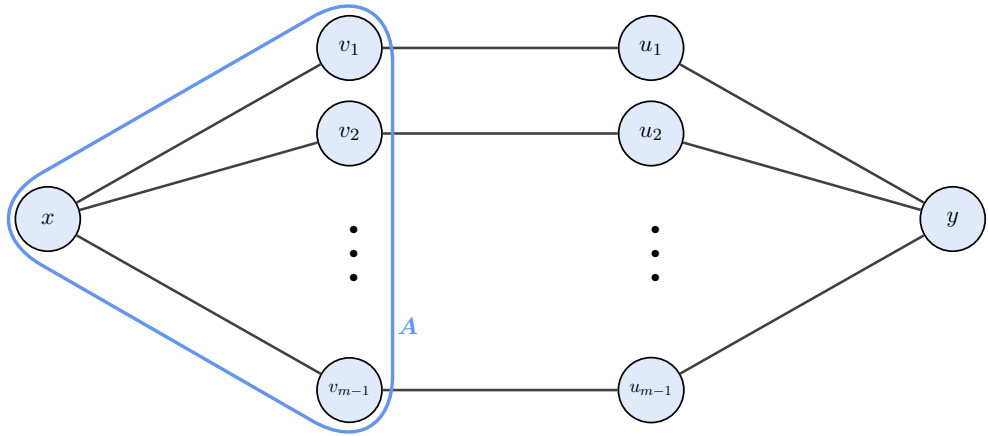
\begin{figure}[h!]
\centering
\scalebox{\valuescale}{
\begin{tikzpicture}[scale = 0.7]
\begin{scope}[every node/.style={circle,minimum size=30pt,thick,draw,fill=cornflowerblue!20}]
    \node (A) at (0,4) {$x$};
    \node (B) at (7,8) {$v_1$};
    \node (C) at (7,6) {$v_2$};
    \node (D) at (7,0) {\scalebox{0.8}{$v_{m-1}$}};
    \node (E) at (14,8) {$u_1$};
    \node (F) at (14,6) {$u_2$} ;
    \node (G) at (14,0) {\scalebox{0.8}{$u_{m-1}$}} ;
    \node (H) at (21,4) {$y$} ;
\end{scope}
\begin{scope}[every node/.style={},
                every edge/.style={draw=darkgray,very thick}]
    \path [-] (A) edge node {} (B);
    \path [-] (A) edge node {} (C);
    \path [-] (A) edge node {} (D);
    \path [-] (B) edge node {} (E);
    \path [-] (C) edge node {} (F);
    \path [-] (D) edge node {} (G);
    \path [-] (E) edge node {} (H);
    \path [-] (F) edge node {} (H);
    \path [-] (G) edge node {} (H);
\end{scope}
\begin{scope}[style={draw=blue, ultra thick}]
    \draw (7.1,3.2) node[rotate = 90](){\scalebox{2}{$\boldsymbol{\ldots}$}};
    \draw (14.1,3.2) node[rotate = 90](){\scalebox{2}{$\boldsymbol{\ldots}$}};
\end{scope}
\begin{scope}[style={draw=cornflowerblue, ultra thick}]
    \draw[rounded corners=15mm] (-2,4)--(8,9.8)--(8,-1.8)--cycle;
    \draw (8.35,1.5) node[text=cornflowerblue](){$\boldsymbol{A}$};
\end{scope}
\end{tikzpicture}
}
\caption{Illustration of a minimal local set of size $n/2$ in $G$.}
\label{fig:bounds3}
\end{figure}

$A$ is a local set, as $A = D \cup Odd_G(D)$ with $D = \{x\}$.
In particular, $A$ is a minimal local set. Indeed, let us assume that there exists a subset $D'$ such that $D' \cup Odd_G(D') \varsubsetneq A$. It cannot contain any of the $v_i$'s as we would then also have $u_i \in Odd_G(D')$ (since $y \notin D'$). So $D'$ can only contain $x$, leading to a contradiction.
\end{proof}

\section{Minimal local sets in a path}
\label{sec:MLSPath}

Let $P_n$ be a path of order $n>2$ on vertices $v_0, \ldots, v_{n-1}$ such that, for any $i$, there is an edge between $v_i$ and $v_{i+1}$. For any $k\in [0, \lceil n/2\rceil-1)$, let $D_k:=\{v_0,v_2,v_4,\ldots, v_{2k}\}$. The local set $L_k$ generated by $D_k$ is $L_k:=D_k\cup Odd_{P_n}(D_k) = D_k\cup \{v_{2k+1}\}$. We show that $L_k$ is minimal by inclusion: by contradiction, assume there exists a non-empty $D\subseteq L_k$, $D\neq D_k$ and $D\cup Odd_{P_n}(D) \subsetneq L_k$. First notice that $v_{2k+1}\notin D$, otherwise $v_{2k+2}\in Odd_{P_n}(D)$ (notice that $2k+2<n$, so $v_{2k+2}$ is actually a vertex of the path). So $D\subsetneq D_k$, as a consequence there exists  $r<k$  s.t. $v_{2r+1}$ has exactly one neighbor in $D$, implying $v_{2r+1}\in Odd_{P_n}(D)$. This is a contradiction since $v_{2r+1}\notin L$. 

As a consequence, for any $k\in [0, \lceil n/2\rceil-1)$, $L_k$ is a minimal local set of size $k+2$. Notice that we have exhibited minimal local sets of any size from $2$ to $\lceil n/2\rceil$, but this is not a complete description of all minimal local sets of the path, for instance $\{v_1,v_2,v_3\}$ (when $n>4$) and $\{v_0, v_2, v_4, \ldots, v_{n-1}\}$ (when $n=1\bmod 2$) are also minimal local sets.

\section{Complete proof of Proposition \ref{prop:exp_number_MLS}}
\label{app:proof_exp_number_MLS}

\expnumberMLS*

\begin{proof}

Using \cref{prop:cutrank_prop2}, any set of $V$ of size $\lfloor n/2 \rfloor +1$ is not a full-cut-rank set, so it contains at least one minimal local set (see \cref{prop:characMLS} for the link between (minimal) local sets and the cut-rank function). Notice that a minimal local set of size at least $rn$ is contained it at  most $\binom{(1-r)n}{ \lceil n/2 \rceil -1}$ subsets of size  $\lfloor n/2 \rfloor +1$, as one has to choose $\lceil n/2 \rceil -1$ vertices that will not be in such a set, among the $(1-r)n$ ones that are not already in the minimal local set. As a consequence, a counting argument implies:

$$\#\{\text{minimal local sets in } G\} \gs \binom{n}{\lfloor n/2 \rfloor +1}\bigg/\binom{(1-r)n }{\lceil n/2 \rceil -1}$$

\begin{itemize}
\item If $n=2k$, 
\begin{align*}
\#\{\text{minimal local sets in } G\} &\gs \binom{2k}{k+1}\binom{2(1-r)k}{k-1}^{-1}\\
&=\frac{(1-2r)k+1}{k+1} \binom{2k}{k}\binom{2(1-r)k}{k}^{-1}\\
&\gs{(1-2r)} \binom{2k}{k}\binom{2(1-r)k}{k}^{-1}\\
\end{align*}

Now, using the standard inequalities $ \binom{2k}{k} \gs \frac{2^{2k-1}}{\sqrt{\pi k}}$ and $\binom{dk}{k} \leqslant 2^{dkH_2(\frac1{d})}$, where $H_2(x) = -x\log_2(x) - (1-x)\log_2(1-x)$ is the binary entropy, we get:
\begin{align*}
\#\{\text{minimal local sets in } G\} &\gs \frac{1-2r}{2\sqrt{k\pi}} ~2^{2k\left(1-(1-r)H_2\big(\frac{1}{2(1-r)}\big)\right) }\\
&= \frac{1-2r}{\sqrt{2n\pi}}~2^{n\left(1-(1-r)H_2\big(\frac{1}{2(1-r)}\big)\right) }\\
&\gs \frac{1-2r}{3\sqrt{n}}~2^{n\left(1-(1-r)H_2\big(\frac{1}{2(1-r)}\big)\right) }
\end{align*}

\item  If $n=2k-1$, 
\begin{align*}
\#\{\text{minimal local sets in } G\} &\gs \binom{2k-1}{k}\binom{(1-r)(2k-1)}{k-1}^{-1}\\
&=\frac{k}{2k} \!\binom{2k}{k}\!\!\left(\!\frac{k}{(1-r)(2k-1)-k+1}\binom{(1-r)(2k-1)}{k}\!\!\right)^{\!-1}\\
&=\frac{k(1-2r)+r}{2k} \binom{2k}{k}\binom{(1-r)(2k-1)}{k}^{-1}
\end{align*}

We have $\binom{(1-r)(2k-1)}{k}\leqslant 2^{(1-r)(2k-1)H_2(\frac k {(1-r)(2k-1)})}$. 
Since $\frac 12 \leqslant \frac{1}{2(1-r)} \leqslant  \frac{k}{(1-r)(2k-1)}$ 
and $H_2$ is decreasing on 
$[1/2,1]$, we have 
$H_2(\frac k {(1-r)(2k-1)})\leqslant H_2(\frac 1 {2(1-r)})$. 

As a consequence, we get:
\begin{align*}
\#\{\text{minimal local sets in } G\} 
&\gs\frac{k(1-2r)+r}{2k} \binom{2k}{k} 2^{-(1-r)(2k-1)H_2(\frac 1 {2(1-r)})} \\
&\gs \frac{k(1-2r)+r}{2k}~\frac{2^{2k-1}}{\sqrt{\pi k}}~2^{-(1-r)(2k-1)H_2\big(\frac{1}{2(1-r)}\big)}\\
&= \frac{(n+1)(1-2r)+2r}{n+1}~\frac{2^{n}}{\sqrt{2\pi (n+1)}}~2^{-(1-r)nH_2\big(\frac{1}{2(1-r)}\big)}\\
&\gs (1-2r)\frac{2^{n}}{\sqrt{2\pi (n+1)}}~2^{-(1-r)nH_2\big(\frac{1}{2(1-r)}\big)}\\
&\gs\frac{1-2r}{3\sqrt{n}}~2^{n\left(1-(1-r)H_2\big(\frac{1}{2(1-r)}\big)\right) }
\end{align*}
The last inequality is due to fact that $3\sqrt n \gs \sqrt{2\pi(n+1)}$ for any $n>2$. 
\end{itemize}
\end{proof}

\section{Minimal local sets in $K_{k,k} \Delta M_k$}
\label{app:bipartite_matching}

Let $G_k = K_{k,k} \Delta M_k$, of order $2k$. $G_k$ is defined as the  symmetric difference of a complete bipartite graph and a matching. More precisely, we write $G_k = (V_1\cup V_2, E)$ with $V_1=\{u_0, \ldots, u_{k-1}\}$ and $V_2=\{v_0, \ldots, v_{k-1}\}$, such that for any $i$ and $j$, $(u_i,u_j)\notin E$, $(v_i,v_j)\notin E$ and $(u_i,v_j)\in E$ if $i \neq j$. Let $\omega \se [0,k-1]$ be a subset of the integers between 0 and $k-1$ such that $|\omega| = 1 \bmod 2$. The local set $L_w$ generated by $D_\omega = \bigcup_{i\in \omega}\{u_i\}$ is $L_\omega:=D_\omega\cup Odd_{G_k}(D_\omega)=D_\omega\cup \bigcup_{i\in [0,k-1] \sm \omega}\{v_i\}$. We show that $L_\omega$ is minimal by inclusion: let $D \se L_\omega$ such that $D\cup Odd_{G_k}(D) \se L_\omega$, we prove that $D \cup Odd_{G_k}(D)$ is either $\emptyset$ or $L_\omega$.

\begin{itemize}[topsep=0.2em,itemsep=0.2em]
    \item Assume $|D \cap V_1| = 0 \bmod 2$. If $D \cap V_1 \neq \emptyset$, i.e. there exists $u_i \in D \cap V_1$, then $v_i \in Odd_{G_k}(D)$, contradicting $v_i\notin L_\omega$. Thus, $D \cap V_1 = \emptyset$.
    \item Assume $|D \cap V_1| = 1 \bmod 2$. If $D \cap V_1 \neq D_\omega \cap V_1$, i.e. there exists $u_i \in D_\omega \sm D$, then $v_i \in Odd_{G_k}(D)$. Thus, $D \cap V_1 = D_\omega \cap V_1$. 
    \item Assume $|D \cap V_2| = 0 \bmod 2$. If $D \cap V_2 \neq \emptyset$, i.e. there exists $v_i \in D \cap V_2$, then $u_i \in Odd_{G_k}(D)$. Thus, $D \cap V_2 = \emptyset$.
    \item Assume $|D \cap V_2| = 1 \bmod 2$. If $D \cap V_2 \neq D_\omega \cap V_2$, i.e. there exists $v_i \in Odd_{G_k}(D_\omega) \sm D$, then $u_i \in Odd_{G_k}(D)$. Thus, $D \cap V_2 = D_\omega \cap V_2$. 
\end{itemize}

To sum up, $D$ is either $\emptyset$, $D_\omega$, $Odd_{G_k}(D_\omega)$, or $L_\omega$. In any case, $D \cup Odd_{G_k}(D)$ is either $\emptyset$ or $L_\omega$ . As a consequence, for any $\omega \se [0,k-1]$ such that $|\omega| = 1 \bmod 2$, $L_w$ is a minimal local set of size $k$. These minimal local sets are distinct, as $L_\omega \cap V_1 = \omega$. Thus, $G_k$ contains at least $2^{k-1}$ minimal local sets.

\section{Explicit algorithm to find a MLS cover}
\label{app:algorithm}

\begin{algorithm}[H]
    \caption{Find an MLS cover}\label{alg:MLS_cover}
    \SetKwInOut{Input}{Input}
    \SetKwInOut{Output}{Output}
    \Input{A graph $G=(V,E)$.}
    \Output{An MLS Cover $\calM$ for $G$.}
    $\calM \gets \{\}$\;
    \While{some vertex $a \in V$ is not covered, i.e. not in a minimal local set in $\calM$}{
        $A \gets \emptyset$\;    
        \While{$A \cup \{a\}$ is full cut-rank}{
            $B \gets V \sm (A\cup\{a\})$\;
            \While{$|B| > |A|+1$}{
                Find $b \in B$ such that $\cutrk(B \sm \{b\}) \gs |A|+1$\;
                $B \gets B \sm \{b\}$\;
            }
            $A \gets B$\;
        }
        \While{$\exists b \in A$ such that $(A \sm \{b\})\cup\{a\}$ is not full cut-rank}{
            $A \gets A \sm \{b\}$\;
        }
        $\calM \gets \calM\cup\{A \cup \{a\}\}$\;

    }
\end{algorithm}

\section{Minimal local sets have either 1 or 3 generators}
\label{app:1or3}

Minimal local sets have either 1 or 3 generators, and in the latter case, there are of even size. This was proved in \cite{VandenNest05} using the quantum stabilizer formalism. Here we provide a graph-theoretical proof. Recall that $\Delta$ denotes the symmetric difference on vertices. $\Delta$ is commutative and associative, satisfies $A \Delta \emptyset = A$, $A \Delta A = \emptyset$ and $Odd_G(D_0\Delta D_1) =  Odd_G(D_0) \Delta Odd_G(D_1)$.

\begin{proposition}
    Given a minimal local set $L$, only two cases can occur:
    \begin{itemize}[topsep=0.2em,itemsep=0.2em]

        \item $L$ has exactly one generator.
        \item $L$ has exactly three (distinct) generators, of the form $D_0$, $D_1$, and $D_0 \Delta D_1$. This can only occur if $|L|$ is even.
    \end{itemize}
    In other words, $|L|-\cutrk(L)=$ 1 or 2, and the latter can only occur if $|L|$ is even.
\label{prop:MLS2cases}
\end{proposition}
    
\begin{proof}
First, let us prove that $L$ has either 1 or 3 generators.

Suppose there exist $D_0$ and $D_1$ ($D_0 \neq D_1$) such that $D_0\cup Odd_G(D_0) = L$ and $D_1\cup Odd_G(D_1) = L$. We know that $D_0 \Delta D_1 \neq \emptyset$. Besides, $(D_0 \Delta D_1)\cup Odd_G(D_0 \Delta D_1) = (D_0 \Delta D_1)\cup (Odd_G(D_0) \Delta Odd_G(D_1)) \se D_0 \cup D_1\cup Odd_G(D_0) \cup Odd_G(D_1) \se L$. And, as $L$ is a minimal local set, $D_0 \Delta D_1$ is a generator of $L$.

We will first prove that $L = D_0 \cup D_1$. 
Suppose that there exists a vertex $v \in L \sm (D_0 \cup D_1)$.
Then $v \in Odd_G(D_0)$ and $v \in Odd_G(D_1)$ so $v \in Odd_G(D_0) \cap Odd_G(D_1)$.
So $v \notin Odd_G(D_0) \Delta Odd_G(D_1) = Odd_G(D_0 \Delta D_1)$.
So, as $(D_0 \Delta D_1)\cup Odd_G(D_0 \Delta D_1) = L$, $v \in D_0 \Delta D_1 \se D_0 \cup D_1$, leading to a contradiction.
    
Now, we will prove that the three subsets, $D_0$, $D_1$, and $D_0 \Delta D_1$, are the only generators of $L$.
Suppose that there exists $D_2 \se L$ such that $D_0 \neq D_2$, $D_1 \neq D_2$, and $D_2 \cup Odd_G(D_2) = L$. Same as before, we can prove that $(D_0 \Delta D_2)\cup Odd_G(D_0 \Delta D_2) = L$ and $(D_1 \Delta D_2)\cup Odd_G(D_1 \Delta D_2) = L$. We also have that $L = D_0 \cup D_2$ and $L = D_1 \cup D_2$. Then $D_0 \cup D_1 = D_0 \cup D_2 = D_1 \cup D_2$, implying $D_1 \sm D_0 \se D_2$ and $D_0 \sm D_1 \se D_2$.
So $D_0 \Delta D_1 = (D_0 \sm D_1) \cup (D_1 \sm D_0) \se D_2$.\\
Suppose there exists $v \in D_2$ such that $v \in L \sm (D_0 \Delta D_1) = D_0 \cap D_1$. Then $v \in D_0 \cap D_2$, so $v \notin D_0 \Delta D_2$, and $v \in Odd_G(D_0 \Delta D_2)$.
For the same reason, $v \in Odd_G(D_1 \Delta D_2)$.
So $v \notin Odd_G(D_0 \Delta D_2) \Delta Odd_G(D_1 \Delta D_2) = Odd_G(D_0 \Delta D_1)$.
So $v \in D_0 \Delta D_1$, leading to a contradiction.
Therefore, we proved that $D_2 = D_0 \Delta D_1$.\\

Second, let us prove that the case where $L$ has 3 generators can only occur when $|L|$ is even.

All calculations are done modulo 2: "$\equiv$" means "is of same parity as".
Suppose that $L$ has three generators  $D_0$, $D_1$ and $D_0 \Delta D_1$. Then $L = D_0 \cup D_1 = \left(D_0 \sm D_1\right) \uplus \left(D_1 \sm D_0\right) \uplus \left(D_0 \cap D_1\right)$.
Given a subset $K$ of $V(G)$ and a vertex $v \in V$, let $\delta_{K}(v)$ denote the degree of $v$ inside $K$, i.e. $\delta_{K}(v) = |N_G(v)\cap K|$. Observe that for any $K \se V$, $\sum_{v\in K}\delta_{K}(v) = 2 \times \text{number of edges in $K$}\equiv 0$.
    
Given $v \in D_0 \sm D_1$, we know that $v \in Odd_G(D_1)$. Then $\delta_{L}(v) \equiv \delta_{D_1}(v) + \delta_{D_0 \sm D_1}(v) \equiv 1 + \delta_{D_0 \sm D_1}(v)$.
So $\sum_{v\in D_0 \sm D_1}\delta_{L}(v) \equiv |D_0 \sm D_1| + \sum_{v\in D_0 \sm D_1}\delta_{D_0 \sm D_1}(v) \\ \equiv |D_0 \sm D_1|$.
Similarly, we have $\sum_{v\in D_1 \sm D_0}\delta_{L}(v) \equiv |D_1 \sm D_0|$ and $\sum_{v\in D_0 \cap D_1}\delta_{L}(v) \\ \equiv |D_0 \cap D_1|$.   
Summing all three cases, we have: $0 \equiv \sum_{v\in L}\delta_{L}(v) \equiv \sum_{v\in D_0 \sm D_1}\delta_{L}(v) \\+ \sum_{v\in D_1 \sm D_0}\delta_{L}(v) + \sum_{v\in D_0 \cap D_1}\delta_{L}(v) \equiv |D_0 \sm D_1| + |D_1 \sm D_0| + |D_0 \cap D_1| \equiv |L|$.
So $|L|$ is even.\\

The interpretation using the cut-rank function is explained in \cref{sec:alt_def}.
\end{proof}

\end{document}